\documentclass{article}
\usepackage{a4wide}
\usepackage{amsmath}
\usepackage{amssymb}
\usepackage{amsthm}
\newtheorem{thm}{Theorem}
\newtheorem{prop}{Proposition}

\newtheorem{cor}{Corollary}
\newtheorem{lemma}{Lemma}

\newcommand{\bbr}{{\mathbb R}}
\newcommand{\bbs}{{\mathbb S}}

\newcommand{\p}{\hat{p}}
\newcommand{\q}{\hat{q}}
\newcommand{\n}{\hat{n}}

\usepackage{authblk}

\begin{document}
\title{Bianchi I solutions of the Einstein-Boltzmann system with a positive cosmological constant}

\author[1]{Ho Lee\footnote{holee@khu.ac.kr}}
\author[2]{Ernesto Nungesser\footnote{ernesto.nungesser@icmat.es}}

\affil[1]{Department of Mathematics and Research Institute for Basic Science, Kyung Hee University, Seoul, 02447, Republic of Korea}
\affil[2]{Instituto de Ciencias Matem\'{a}ticas (CSIC-UAM-UC3M-UCM), 28049 Madrid, Spain}

\maketitle

\begin{abstract}
In this paper we study the future global existence and late-time behaviour of the Einstein-Boltzmann system with Bianchi I symmetry and a positive cosmological constant $\Lambda>0$. For the Boltzmann equation we consider the scattering kernel of Israel particles which are the relativistic counterpart of Maxwellian particles. Under a smallness assumption on initial data in a suitable norm we show that solutions exist globally in time and isotropize at late times. \end{abstract}

\section{Introduction}
When describing cosmological solutions at a macroscopic level a usual choice for the matter model is a fluid. The pressure is then a function of the energy density, and the relation between these two quantities is often linear. However, a more refined description is obtained via kinetic theory. One of the most important models is the Einstein-Boltzmann system, and this system is very interesting from the mathematical and physical points of view; the Boltzmann equation is a bridge between micro- and macroscopic laws and is crucial for the understanding of many-body physics. In some sense it is also a bridge between classical and quantum physics. In this paper we are interested in the Einstein-Boltzmann system and study the late-time behaviour of it.

We assume the expansion of the Universe as a fact and consider a positive cosmological constant so that the expansion is accelerating. The main objective of the present paper is to improve the previous results \cite{LeeNungesser171,LeeNungesser172}. In \cite{LeeNungesser171} we showed that Bianchi I solutions to the Einstein-Boltzmann system exist globally in time, but imposed artificial restrictions on the scattering kernel. In \cite{LeeNungesser172} we showed that classical solutions to the Boltzmann equation exist globally in time, but assumed that an isotropic spacetime is given with spatially flat geometry and exponentially growing scale factor. In the present paper we continue this line of research considering the coupled Einstein-Boltzmann system with Bianchi I symmetry. Future global existence is shown, and the asymptotic behaviour of the distribution function and the relevant geometric quantities is obtained with specific decay rates. It is shown that the spacetime tends to de Sitter spacetime at late times. To the best of our knowledge it is the first result of this type concerning the coupled Einstein-Boltzmann system with a scattering kernel which is physically well motivated.

The paper is structured as follows. Before coming to the next section, we will introduce some notations and in particular define what will be the relevant norm. We finish the introduction considering isotropic solutions to the coupled Einstein-Boltzmann system which is easily established based on \cite{LeeNungesser172}. This gives the reader an idea of what will be done in the following sections for the more complicated Bianchi I case, where the spacetime is still homogeneous but anisotropic. In Section 2 we present our main results. After that we collect the main equations in Section 3 and some basic estimates in Section 4. The key estimates developed in this paper are then established in Section 5. Finally in Section 6 we explain how the different estimates are combined to prove the main theorem and present some ideas on how to continue this line of research.

\subsection{Notations}\label{sec_notations}
Let $A=(a_1,\cdots,a_m)$ be an $m$-tuple of integers between $1$ and $3$. We write $\partial^A=\partial^{a_1}\cdots\partial^{a_m}$, where $\partial^a=\partial /\partial p_a$ is the partial derivative with respect to $p_a$ for $a\in\{1,2,3\}$ with $|A|=m$ the total order of differentiation. Indices are lowered and raised by the metric $g_{\alpha\beta}$ and $g^{\alpha\beta}$, respectively, and the Einstein summation convention is assumed. Greek letters run from $0$ to $3$, while Latin letters from $1$ to $3$, and momentum variables without indices denote three dimensional vectors, for instance we write
\[
p=(p^1,p^2,p^3),\quad p_*=(p_1,p_2,p_3).
\]

We consider an orthonormal frame $e_\mu(=e_\mu^\alpha E_\alpha)$, i.e. $g_{\alpha\beta}e^\alpha_\mu e^\beta_\nu=\eta_{\mu\nu}$, where $e_0=\partial_t$ and $\eta_{\mu\nu}$ denotes the Minkowski metric. Momentum variables can be written as $p^\alpha E_\alpha=\p^\mu e_\mu$ where we use a hat to indicate that the momentum is written in an orthonormal frame. Note that
\[
p^\alpha=\p^\mu e_\mu^\alpha,\quad \p_\mu(=\eta_{\mu\nu}\p^\nu)=p_\alpha e^\alpha_\mu,
\]
where the Minkowski metric applies in an orthonormal frame. For partial derivatives in an orthonormal frame, we use hats in a similar way, and the derivatives with respect to $p_a$ and $\p_a$ are related to each other as $\partial^a=e^a_b\hat{\partial}^b$. For a multi-index $A=(a_1,\cdots,a_m)$ we have
\[
\partial^A=e^A_B\hat{\partial}^B,
\]
where $e^A_B=e^{a_1}_{b_1}\cdots e^{a_m}_{b_m}$ for $B=(b_1,\cdots,b_m)$.

We also consider the usual $l^2$-norm: for a three-dimensional vector $v$, we define
\[
|v|=\sqrt{\sum_{i=1}^3(v^i)^2},
\]
and note that $|\p|^2=\eta_{ab}\p^a\p^b$. With this notation we define the weight function:
\[
\langle p_*\rangle=\sqrt{1+|p_*|^2},
\]
and note that it is different from $p^0$ in general. With this weight function we define the norm of a function $f=f(t,p_*)$ as follows: for a non-negative integer $N$,
\begin{align*}
\|f(t)\|^2_{k,N}=\sum_{|A|\leq N}\|\partial^A f(t)\|_k^2,\quad \|f(t)\|_k^2=\int_{\bbr^3}\langle p_*\rangle^{2k}e^{p^0(t)}|f(t,p_*)|^2dp_*,
\end{align*}
where $k$ is a positive real number.

\subsection{Isotropic solutions in the coupled Einstein-Boltzmann system}
Before coming to the Bianchi I case, let us establish the results for the coupled Einstein-Boltzmann system with FLRW symmetry in the case of spatially flat geometry based on the previous result \cite{LeeNungesser172}. We consider the metric
\begin{align*}
^{(4)}g=-dt^2+g,\quad g=R^2 (dx^2 + dy^2+dz^2),
\end{align*}
where $R=R(t)>0$ is the scale factor. Thus the spatial metric is $g_{ab}=R^2 \delta_{ab}$, so that $k_{ab}=\dot{g}_{ab}/2=R \dot{R} \delta_{ab}$, where the dot denotes the derivative with respect to time. As a consequence the Hubble variable is $H=k_{ab} g^{ab}/3 = \dot{R}/R$. In contrast to the Bianchi I case, which we will consider in the following sections, there is no shear, i.e. $\sigma_{ab}=k_{ab}-H g_{ab}=0$. For the isotropic case we also assume that the distribution function is invariant under rotations of the momenta. The collision kernel will be described in the next section, but for the moment we only notice that it is invariant under rotations; it does not depend on the scattering angle, and the other variables are invariant under the group $SO(3)$, for example $p^0=(1+R^{-2}|p_*|^2)^{-1/2}$ which depends only on the length of $p_*$. Also the quantity $s$ is invariant under the group $SO(3)$ since $s=(p^0+q^0)^2-R^{-2}|p_*-q_*|^2$. As a consequence the Boltzmann equation preserves the isotropy. Note that we have here an example of a collision kernel satisfying (4.4) of \cite{NTRW}. The governing equations in a FLRW spacetime with a cosmological constant $\Lambda$, which will be considered to be positive, are given as follows:
\begin{align*}
\frac{{\dot{R}}^{2}}{R^{2}}&=\frac{\rho+\Lambda}{3},\\
\frac{3\ddot{R}}{R}&=-\frac{\rho+3P}{2}+\Lambda,
\end{align*}
which are called the Friedmann equations. Here $\rho$ and $P$ are the energy density and the pressure, respectively, which are given by
\begin{align*}
\rho &=  R^{-3} \int_{\bbr^3} f p^0dp_*,\\
P &= R^{-5} \int_{\bbr^3} f \frac{|p_*|^2}{3p^0}dp_*.
\end{align*}
We are interested in solutions having an asymptotic behaviour concerning the distribution function as in \cite{LeeNungesser172}. We thus assume that 
\begin{align*}
f(t,p_*) \leq C_f (1+\vert p_* \vert^2)^{-\frac{k}{2}} e^{-\frac12 p^0},
\end{align*}
for some large $k$, or in an orthonormal frame
\begin{align*}
\hat{f}(t,\p) \leq C_f (1+R^2 \vert \hat{p} \vert^2)^{-\frac{k}{2}} e^{-\frac12 p^0}.
\end{align*}
We assume that our Universe is initially expanding, i.e. $H(t_0)>0$ and non-empty, i.e. $f(t_0)\neq0$. The vacuum case corresponds to the well-known Kasner solutions. Then, from the first Friedmann equation we have $H > \gamma=(\Lambda/3)^{1/2}$ for all times. Concerning the energy density, we have
\begin{align*}
\rho&\leq C_f  R^{-3} \int  (1+ \vert p_* \vert^2)^{-\frac{k}{2}} e^{-\frac12\sqrt{1+ R^{-2} \vert p_* \vert^2}}\sqrt{1+ R^{-2} \vert p_* \vert^2}\,  dp_* \\
&\leq  4\pi C_f R^{-3} \int_0^\infty  r^2 (1+ r^2)^{-\frac{k}{2}} e^{-\frac12 \sqrt{1+ R^{-2} r^2}}\sqrt{1+ R^{-2} r^2} \, dr \leq C C_f R^{-3},
\end{align*}
where we used the fact that $R$ is bounded from below since $H$ is. Similarly we obtain that
\begin{align*}
P \leq C C_f R^{-5}.
\end{align*}
Thus $\rho$ and $P$ are bounded. Applying standard arguments we obtain the global existence of $R$ given $f$ satisfying the asymptotic behaviour mentioned above. Moreover, we can easily obtain an improved estimate for $H$.  Using the lower bound obtained in \cite{LeeNungesser172} for $R$, namely $R\geq e^{\gamma t}$, and the asymptotic behaviour of $f$, we have that the energy density decays exponentially:
\begin{align*}
\rho \leq  C C_f  R^{-3} \leq C C_f e^{-3\gamma t},
\end{align*}
and $P\leq CC_f e^{-5\gamma t}$ for the pressure. From the Friedmann equations we have that
\begin{align*}
\dot{H}= \gamma^2 -H^2-\frac{1}{6}(\rho+3P)\leq\gamma^2 -H^2,
\end{align*}
and we derive
\begin{align*}
\frac{d}{dt}(H-\gamma)\leq -(H-\gamma)(H+\gamma)\leq -2\gamma (H-\gamma),
\end{align*}
which shows that $H$ converges to $\gamma$ exponentially: 
\begin{align*}
H-\gamma =O(e^{-2\gamma t}).
\end{align*}
Let us define $R= e^{\gamma t} \bar{R}$ and use the definition of $H=\dot{R}/R$ to obtain
\begin{align*}
\dot{\bar{R}}=(H-\gamma) \bar{R}.
\end{align*}
Using the estimate of $H-\gamma$ we conclude that $\bar{R}$ is bounded, and as a consequence we obtain
\begin{align*}
\dot{\bar{R}}=O(e^{-2\gamma t}).
\end{align*}
This means that $\bar{R}$ will in fact converge to the expression given by
\begin{align*}
\bar{R}_\infty= \bar{R}(t_0)+ \int^\infty_{t_0} (H-\gamma) \bar{R}\, dt.
\end{align*}
Putting things together we obtain that $R$ converges exponentially to an exponentially growing function:
\begin{align*}
R= e^{\gamma t} (\bar{R}_\infty + O(e^{-2\gamma t})).
\end{align*}
We now combine this with the result of \cite{LeeNungesser172} to obtain the result for the coupled case. Suppose that initial data $R(t_0)$, $\dot{R}(t_0)$, and $f(t_0)$ are given, and define an iteration $\{R_n\}$ and $\{f_n\}$ as follows. Let $R_0= e^{\gamma t} \bar{R}(t_0)$, and note that $R_0$ satisfies the condition of Theorem 1 of \cite{LeeNungesser172}. As a consequence there exists a small positive $\varepsilon_0$ such that if $\Vert f(t_0) \Vert^2_{k,N} <\varepsilon_0$, where the norm is defined in Section \ref{sec_notations}, then there exists a unique non-negative classical solution $f_0$ to the Boltzmann equation in a given spacetime with scale factor $R_0$. Note that the solution satisfies $\sup_{0\leq t<\infty}\|f_0(t)\|^2_{k,N}\leq C\varepsilon_0$ for some constant $C>0$. Suppose now that $f_n$ is given such that $\sup_{0\leq t<\infty}\Vert f_n(t) \Vert^2_{k,N} \leq C\varepsilon_0$ with $\Vert f(t_0) \Vert^2_{k,N} <\varepsilon_0$. This means that we have the desired asymptotic behaviour for $f_n$, hence $R_{n+1}$ exists globally in time. As a result we have an exponential growth for the scale factor using Lemma 1 of \cite{LeeNungesser172}. Applying Theorem 1 of \cite{LeeNungesser172} again we obtain $f_{n+1}$, and this completes the iteration. We have uniform bounds for all the relevant quantities and thus can take the limit up to a subsequence to obtain classical solutions to the coupled equations with the asymptotic behaviour for $f$ and $R$ as described above. To summarize, we briefly considered here the isotropic case as an introduction and obtained certain asymptotic behaviours for the metric and the distribution function. For the Bianchi I case the procedure will be much similar to the isotropic case, but below neither the metric nor the distribution function will be isotropic.

\section{Main results}
We state the main theorem of the present paper.
\begin{thm}\label{Thm}
Consider the Einstein-Boltzmann system \eqref{evolution1}--\eqref{constraint2} with Bianchi I symmetry and a positive cosmological constant. Suppose that the assumption on the scattering kernel holds and the Hubble variable is initially given as $H(t_0)<(7/6)^{1/2}\gamma$ with $\gamma=(\Lambda/3)^{1/2}$. Let $g_{ab}(t_0)$, $k_{ab}(t_0)$, and $f(t_0)$ be initial data of the Einstein-Boltzmann system satisfying the constraints \eqref{constraint1}--\eqref{constraint2} such that $\|f(t_0)\|_{k+1/2,N}$ is bounded for $k>5$ and $N\geq 3$. Then, there exists a small $\varepsilon>0$ such that if $\|f(t_0)\|_{k+1/2,N}<\varepsilon$, then there exists a unique classical solutions $g_{ab}$, $k_{ab}$, and $f$ to the Einstein-Boltzmann system corresponding to the initial data. The solutions exist globally in time, the spacetime is geodesically future complete, and the distribution function $f$ is nonnegative. Moreover, there exist constant matrices $\mathcal{G}_{ab}$ and $\mathcal{G}^{ab}$ such that 
\begin{align*}
&H=\gamma+O(e^{-2\gamma t}),\\
&\sigma^{ab}\sigma_{ab}=O(e^{-2\gamma t}),\\
&g_{ab}=e^{2\gamma t}\Big(\mathcal{G}_{ab}+O(e^{-\gamma t})\Big),\\
&g^{ab}=e^{-2\gamma t}\Big(\mathcal{G}^{ab}+O(e^{-\gamma t})\Big),
\end{align*}
and the distribution function $f$ satisfies in an orthonormal frame
\[
\hat{f}(t,\hat{p})\leq C \varepsilon (1+e^{2\gamma t} |\hat{p}|^2)^{-\frac12 k} e^{-\frac12 p^0},
\]
where $C$ is a positive constant. 
\end{thm}

The proof is given in the last section. From the control of the main quantities one can then obtain estimates of related quantities collected in the corollary below. The Kasner exponents or shape parameters tend all to $1/3$, the deceleration parameter tends to the expected value, and we have a dust-like behaviour at late times. The details of this can be found in the Section \ref{sec_einstein} about basic estimates concerning the Einstein part.

\begin{cor}
Let $g_{ab}$, $k_{ab}$, and $f$ be the solutions obtained in the previous theorem. Then, we also obtain the following estimates:
\begin{align*}
&s_i= \frac13 + O(e^{-\gamma t}),\\
&d= -1 + O(e^{-2 \gamma t}),\\
&\rho=O(e^{-3\gamma t}),\\
&\hat{S}_{ij}=O(e^{-5\gamma t}), \\
&\frac{\hat{S}_{ij}}{\rho}= O(e^{-2\gamma t}).
\end{align*}
\end{cor}

\section{The Einstein-Boltzmann system with Bianchi I symmetry in the case of Israel particles}
In this paper we are interested in the spacetime with Bianchi type I symmetry. We follow the sign conventions of \cite{Hans} and also refer to this book for background on the Einstein equations with Bianchi symmetry. Concerning the relativistic kinetic theory we refer to \cite{CercignaniKremer}. Using a left-invariant frame $E_a$ with $\xi^a$ its duals, the metric of a Bianchi spacetime can be written as
\[
^{(4)}g=-dt\otimes dt+g,\quad g=g_{ab}\,\xi^a\otimes\xi^b.
\]
In the Bianchi I case the evolution equations of metric $g_{ab}$ and second fundamental form $k_{ab}$ are obtained via the Einstein equations and are as follows (cf. (25.17)--(25.18) of \cite{Hans}):
\begin{align}
\dot{g}_{ab}&=2k_{ab},\label{evolution1}\\
\dot{k}_{ab}&=2k^c_{a}k_{bc}-k\,k_{ab}+S_{ab}+\frac{1}{2}(\rho-S)g_{ab}+\Lambda g_{ab},\label{evolution2}
\end{align}
where a dot denotes derivation with respect to the time, $k=g^{ab}k_{ab}$, $k^c_a=g^{cd}k_{da}$, and $\det g$ is the determinant of the matrix $(g_{ab})$. 
Moreover $\rho$ and $S$ come from the energy-momentum tensor $T_{\alpha\beta}$ which, since we use a kinetic picture, is defined by
\[
T_{\alpha\beta}=(\det g)^{-\frac12}\int_{\bbr^3}f(t,p_*)\frac{p_\alpha p_\beta}{p^0}dp_*,
\]
and $\rho=T_{00}$, $S_{ab}=T_{ab}$, and $S=g^{ab}S_{ab}$. The evolution equations \eqref{evolution1}--\eqref{evolution2} are coupled to the Boltzmann equation:
\begin{align}
\partial_tf&=(\det g)^{-\frac12}\int_{\bbr^3}\int_{\bbs^2}v_M\sigma(h,\theta)\Big(f(p_*')f(q_*')-f(p_*)f(q_*)\Big)d\omega dq_*\label{boltzmann} \\ 
&=:Q(f,f)=Q_+(f,f)-Q_-(f,f). \nonumber
\end{align}
The $Q(f,f)$ is called the collision operator, where $Q_{\pm}(f,f)$ are called the gain and the loss terms, respectively. The M{\o}ller velocity $v_M$ and the relative momentum $h$ are defined for given momenta $p^\alpha$ and $q^\alpha$ by
\[
v_M=\frac{h\sqrt{s}}{4p^0q^0},\quad h=\sqrt{(p_\alpha-q_\alpha)(p^\alpha-q^\alpha)},\quad s=-(p_\alpha+q_\alpha)(p^\alpha+q^\alpha),
\]
where $s$ is called the total energy and satisfy $s=4+h^2$. The post-collision momenta $p_\alpha'$ and $q_\alpha'$ are now given by
\begin{align*}
\left(
\begin{array}{c}
p'^0\\
p'_k
\end{array}
\right)=
\left(
\begin{array}{c}
\displaystyle
p^0+2\bigg(-q^0\frac{n_a e^a_b\omega^b}{\sqrt{s}}+q_ae^a_b\omega^b+\frac{n_ae^a_b\omega^bn_c q^c}{\sqrt{s}(n^0+\sqrt{s})}\bigg)\frac{n_de^d_i\omega^i}{\sqrt{s}}\\
\displaystyle
p_k+2\bigg(-q^0\frac{n_ae^a_b\omega^b}{\sqrt{s}}+q_ae^a_b\omega^b+\frac{n_ae^a_b\omega^bn_cq^c}{\sqrt{s}(n^0+\sqrt{s})}\bigg)
\bigg(g_{ka}e^a_b\omega^b+\frac{n_ae^a_b\omega^bn_k}{\sqrt{s}(n^0+\sqrt{s})}\bigg)
\end{array}
\right),
\end{align*}
and $q_\alpha'=p_\alpha+q_\alpha-p_\alpha'$, where $n^\alpha$ denotes $p^\alpha+q^\alpha$ for simplicity, and $\omega=(\omega^1,\omega^2,\omega^3)\in\bbs^2$ serves as a parameter. The $e^a_b$ are the components of an orthonormal frame, and we recall that these were introduced in Section \ref{sec_notations}. For background on the representation of the post-collision momenta we refer to the Appendix of \cite{LeeNungesser172}. The constraints are given by
\begin{align}
-k_{ab}k^{ab}+k^2&=2\rho+2\Lambda,\label{constraint1}\\
0&=-T_{0a}.\label{constraint2}
\end{align}
In the present paper we assume that the cosmological constant is positive, i.e. $\Lambda>0$. The relevant equations of the Einstein-Boltzmann system with Bianchi I symmetry and $\Lambda>0$ are thus the equations \eqref{evolution1}--\eqref{constraint2}, and we now study the global-in-time properties of it. The quantity $\sigma(h,\theta)$ in the collision operator is called the scattering kernel, where the scattering angle $\theta$ is defined by $\cos\theta=(p^\alpha-q^\alpha)(p'_\alpha-q'_\alpha)/h^2$. In this paper we consider Israel particles:
\[
\sigma(h,\theta)=\frac{4}{h(4+h^2)}\sigma_0(\theta),\label{scat}
\]
where $\sigma_0$ is an arbitrary function of the scattering angle $\theta$. For simplicity we assume that
\[
\sigma_0(\theta)\equiv 1.
\]
Hence, the scattering kernel of our interest is written as $\sigma(h,\theta)=4(hs)^{-1}$.

\section{Basic estimates}
\subsection{Estimates for the Einstein part}\label{sec_einstein}
In this section we enumerate the results which are needed for the Einstein part. It is convenient to introduce the trace-free part of the second fundamental form:
\[
\sigma_{ab}= k_{ab}-\frac13 k g_{ab},
\]
and denote the Hubble variable by
\[
H=\frac13 k.
\]
Apart from estimates for the metric components it is sometimes useful to have estimates for the generalised Kasner exponents or sometimes also called shape parameters, which are defined as the quotient of the eigenvalues of the second fundamental form with respect to the metric and the trace of the second fundamental form. We will denote them by $s_i$. Another useful variable is the deceleration parameter $d$, which is defined as
\[
d=-1  -\frac{\dot{k}}{k^2}.
\]
 
The following is the results for the Einstein equations with a given distribution function $f$ satisfying a certain property. 
\begin{prop}\label{prop_einstein}
Consider a Bianchi I spacetime, which is initially expanding, i.e. $H(t_0)>0$, and let $g_{ab}(t_0)$ and $k_{ab}(t_0)$ be initial data of the evolution equations \eqref{evolution1}--\eqref{evolution2} satisfying the constraints \eqref{constraint1}--\eqref{constraint2}. Suppose that a distribution function $f$ is given and satisfies for some positive $C_f$,
\begin{align}\label{asympfI}
\hat{f}(t,\p)\leq C_f (1+ e^{2\gamma t}|\p|^2)^{-\frac{1}{2}k}e^{-\frac12 p^0},
\end{align}
where $\gamma=(\Lambda/3)^{1/2}$ and $k>5$. Then, the Einstein equations admit global-in-time solutions $g_{ab}$ and $k_{ab}$ which satisfy the following estimates:
\begin{align*}
&H=\gamma+O(e^{-2\gamma t}),\\
&\sigma_{ab}\sigma^{ab}= O(e^{-2\gamma t}),\allowdisplaybreaks\\
&g_{ab}= e^{2\gamma t} \Big(\mathcal{G}_{ab}+ O(e^{-\gamma t})\Big),\\
&g^{ab}= e^{-2\gamma t} \Big(\mathcal{G}^{ab}+ O(e^{-\gamma t})\Big),\allowdisplaybreaks\\
&s_i= \frac13 + O(e^{-\gamma t}),\\
&d= -1 + O(e^{-2 \gamma t}),\allowdisplaybreaks\\
&\rho=O(e^{-3\gamma t}),\\
&\hat{S}_{ij}=O(e^{-5\gamma t}), \\
&\frac{\hat{S}_{ij}}{\rho}= O(e^{-2\gamma t}),
\end{align*}
where $\mathcal{G}_{ab}$ and $\mathcal{G}^{ab}$ are constant matrices.
\end{prop}
\begin{proof}
Given the distribution function $f$ satisfying \eqref{asympfI} we estimate $\rho$ as
\begin{align*}
\rho=\int \hat{f}(t,\p)p^0d\p \leq C_f \int (1+e^{2\gamma t} \vert\hat{p} \vert^2)^{-\frac12 k} e^{-\frac12 p^0}p^0 d\p \leq CC_f e^{-3\gamma t}.
\end{align*}
To estimate $\hat{S}_{ij}$ we use \eqref{asympfI} with the fact that $p^0\geq1$ for massive particles:
\begin{align*}
\hat{S}_{ij}&= \int \hat{f}(t,\p) \frac{\hat{p}_i \hat{p}_j}{p^0} d\p \leq C C_f\int   (1+e^{2\gamma t} \vert\hat{p} \vert^2)^{-\frac12 k} e^{-\frac12 p^0} \vert \hat{p} \vert^2 d\p \\
&\leq C C_fe^{-5\gamma t} \int   (1+|z|^2)^{-\frac12 k} e^{-\frac12 \sqrt{1+ e^{-2\gamma t}|z|^2}} \vert z \vert^2 dz \leq C C_f  e^{-5\gamma t},
\end{align*}
where we used $k>5$ for the last inequality. Global-in-time existence of solutions $g_{ab}$ and $k_{ab}$ is easily obtained by standard arguments. The estimates of $H$, $\sigma_{ab}\sigma^{ab}$, $g_{ab}$, $g^{ab}$, $s_i$, and $d$ are also easily obtained by the same arguments as in \cite{Lee04}. In fact the results of \cite{Lee04} go through since the Vlasov equation is not used at all. So we omit the details and only refer to the proofs of Propositions 3.1 and 3.5--3.7 in \cite{Lee04}. To obtain the estimate of the quotient of $\hat{S}_{ij}$ and $\rho$, we note that the energy density is bounded from below by the zeroth component of the particle current density $N^\alpha$, which is defined as
\begin{align*}
N^\alpha=(\det g)^{-\frac12}\int_{\bbr^3}f(t,p_*)\frac{p^\alpha}{p^0}dp_*.
\end{align*}
This quantity is divergence-free and as a result 
\begin{align*}
\dot{N}^0=-3HN^0.
\end{align*}
Using the estimate for $H$ we obtain that
\begin{align*}
\rho > N^0 > C e^{-3\gamma t},
\end{align*}
and this proves the estimate of $\hat{S}_{ij}/\rho$. 
\end{proof}

Here, we assumed that a distribution function is given and obtained global solutions to the Einstein equations. This result will be used in Section \ref{sec_boltzmann} to obtain solutions to the Boltzmann equation for given $g_{ab}$ and $k_{ab}$. We will define an iteration for the coupled equations. To obtain a solution from the iteration we need boundedness of the following quantity:
\[
F=\frac{\sigma_{ab}\sigma^{ab}}{4H^2},
\]
which is a scaled version of the shear.
\begin{lemma}\label{lem_F}
Let $g_{ab}$ and $k_{ab}$ be the solutions obtained in Proposition \ref{prop_einstein}. If initial data is given as $H(t_0)<(7/6)^{1/2}\gamma$, then $F(t)<1/4$ for all $t\geq t_0$.
\end{lemma}
\begin{proof}
Note that $H$ satisfies the following differential equation:
\[
\dot{H}=-3H^2-\frac{1}{6}S+\frac{1}{2}\rho+\Lambda,
\]
and the constraint equation \eqref{constraint1} is written as
\begin{align}\label{constraint3}
k^2=\frac32\sigma_{ab}\sigma^{ab}+3\rho+3\Lambda.
\end{align}
Since $k=3H$, the differential equation for $H$ is now written as
\[
\dot{H}=-\frac12 \sigma_{ab}\sigma^{ab} -\frac{1}{6}S-\frac12\rho\leq 0,
\]
and this shows that $H$ is decreasing. Together with the previous results we can see that $H$ is monotonically decreasing to the constant $\gamma$, in particular $H\geq \gamma$. We use again the constraint equation \eqref{constraint3} to obtain the following inequality:
\[
\sigma_{ab}\sigma^{ab}=6H^2-2\rho-6\gamma^2\leq 6H^2-6\gamma^2.
\]
Then, the quantity $F$ is estimated as follows:
\[
F=\frac{\sigma_{ab}\sigma^{ab}}{4H^2}\leq\frac{3H^2-3\gamma^2}{2\gamma^2}\leq \frac32\bigg(\frac{H^2(t_0)}{\gamma^2}-1\bigg)<\frac14,
\]
and this completes the proof.
\end{proof}

Finally let us note that with the estimate of $H$ the determinant of the metric is estimated as follows. Since
\begin{align*}
\frac{d (\log \det g)}{dt}= 6H,
\end{align*}
we obtain
\begin{align*}
\det g= O( e^{6\gamma t}).
\end{align*}
This means that we can go from our frame to an orthonormal frame and back with transformation matrices which satisfy $|e^a_b|\leq C e^{\gamma t}$ and $|(e^{-1})^a_b|\leq Ce^{-\gamma t}$. For details on this we refer to \cite{LeeNungesser171}.

\subsection{Basic estimates for momentum variables}\label{sec_basic2}
In this part we collect basic lemmas for the Boltzmann equation. 
\begin{lemma}\label{lem_basic}
The following estimates hold:
\[
s=4+h^2,\quad \frac{|\p-\q|}{\sqrt{p^0q^0}}\leq h\leq |\p-\q|,\quad
s\leq 4p^0q^0,\quad
|\p|\leq p^0.
\]
\end{lemma}
\begin{proof}
We refer to \cite{GuoStrain12,LeeNungesser171} for the proofs. Note that $\p^i$ is the $i$-th component of the momentum in an orthonormal frame.
\end{proof}

\begin{lemma}\label{lem_int}
For any integer $m$, we have
\[
\int_{\bbr^3}(p^0)^me^{-p^0}dp_*\leq C(\det g)^{\frac12},
\]
where the constant $C$ does not depend on $t$.
\end{lemma}
\begin{proof}
By a direct calculation we have
\begin{align*}
\int_{\bbr^3}(p^0)^me^{-p^0}dp_*=(\det g)^{\frac12}\int_{\bbr^3} (p^0)^me^{-p^0}d\p\leq C(\det g)^{\frac12},
\end{align*}
where we used the representation $p^0=\sqrt{1+|\p|^2}$ in an orthonormal frame.
\end{proof}

\begin{lemma}\label{lem_pp'q'}
Let a spatial metric $g$ satisfy the assumption {\sf (A)} of Section \ref{sec_boltzmann}. Then, the following estimate holds:
\[
\langle p_*\rangle\leq C\langle p_*'\rangle\langle q_*'\rangle,
\]
where the constant $C$ does not depend on the metric.
\end{lemma}
\begin{proof}
Since $g^{ab}p_ap_b=|\p|^2$, the assumption {\sf (A)} implies that
\[
\frac{1}{c_0}e^{2\gamma t}|\p|^2\leq |p_*|^2\leq c_0 e^{2\gamma t}|\p|^2.
\]
Then, we have in an orthonormal frame,
\[
\frac{1+|p_*|^2}{(1+|p_*'|^2)(1+|q_*'|^2)}\leq \frac{C(1+e^{2\gamma t}|\p|^2)}{(1+e^{2\gamma t}|\p'|^2)(1+e^{2\gamma t}|\q'|^2)}\leq \frac{C(1+e^{2\gamma t}|\p|^2)}{1+e^{2\gamma t}(|\p'|^2+|\q'|^2)},
\]
where the constant $C$ depends only on the constant $c_0$ given in {\sf (A)}. We now follow the proof of Lemma 4 of \cite{LeeNungesser172}, where the factor $e^{2\gamma t}$ is replaced by $R^2$, and obtain the desired result. We refer to \cite{LeeNungesser172} for more details.
\end{proof}

\begin{lemma}\label{lem_partial_1/p0}
For a multi-index $A$, there exist polynomials ${\mathcal P}$ and ${\mathcal P}_i$ such that
\begin{align*}
&\partial^A \bigg[\frac{1}{p^0}\bigg]=\frac{e^A_B}{(p^0)^{|A|+1}}{\mathcal P}\bigg(\frac{\p}{p^0}\bigg),\\
&\partial^A\bigg[\frac{1}{\sqrt{s}}\bigg]=\frac{e^A_C}{\sqrt{s}}\sum_{i=0}^{|A|}\bigg(\frac{q^0}{s}\bigg)^i \bigg(\frac{1}{p^0}\bigg)^{|A|-i}{\mathcal P}_i\bigg(\frac{\p}{p^0},\frac{\q}{q^0}\bigg),
\end{align*}
where the multi-indices $B$ and $C$ are summed with the polynomials.
\end{lemma}
\begin{proof}
Note that the partial derivatives with respect to $p_a$ and $\p_a$ are related to each other as $\partial^a=e^a_b \hat{\partial}^b$, and for high order derivatives we have
\begin{align}
\partial^A=e^A_B\hat{\partial}^B,\label{relation}
\end{align}
where $e^A_B=e^{a_1}_{b_1}\cdots e^{a_m}_{b_m}$ for $A=(a_1,\cdots,a_m)$ and $B=(b_1,\cdots,b_m)$. From the proof of Lemma 5 of \cite{LeeNungesser172} we have in an orthonormal frame
\begin{align*}
\hat{\partial}^A \bigg[\frac{1}{p^0}\bigg]=\frac{1}{(p^0)^{|A|+1}}{\mathcal P}\bigg(\frac{\p}{p^0}\bigg),
\end{align*}
and the first estimate of the lemma follows from the relation \eqref{relation} with the fact that $|A|=|B|$. The second estimate is also obtained by the same argument, and this completes the proof. For more details we refer to Lemma 5 and 6 of \cite{LeeNungesser172}.
\end{proof}

\begin{lemma}\label{lem_partial_1/n0+s}
For a multi-index $A\neq 0$, there exist polynomials ${\mathcal P}_i$ such that
\[
\partial^A\bigg[\frac{1}{n^0+\sqrt{s}}\bigg]=e^A_B\sum_{i=1}^{|A|}\frac{(q^0)^{|A|}}{(n^0+\sqrt{s})^{i+1}} {\mathcal P}_i\bigg(\frac{\p}{p^0},\frac{\q}{q^0}, \frac{1}{p^0}, \frac{1}{q^0},\frac{1}{\sqrt{s}}\bigg),
\]
where the multi-index $B$ is summed with the polynomials.
\end{lemma}
\begin{proof}
This lemma is also proved by the same argument as in the previous lemma. We apply the relation \eqref{relation} to the proof of Lemma 7 of \cite{LeeNungesser172} and obtain the desired result.
\end{proof}

\begin{lemma}\label{lem_p'}
Consider post-collision momenta $p_*'$ and $q_*'$ for given $p_*$ and $q_*$. For a multi-index $A\neq 0$, we have the following estimate:
\[
|\partial^Ap'_*|+|\partial^Aq'_*|\leq C\max_{a,b}|(e^{-1})^a_b|(\max_{c,d}|e^c_d|)^{|A|}(q^0)^{|A|+4},
\]
where the constant $C$ does not depend on $p_*$.
\end{lemma}
\begin{proof}
Note that the post-collision momentum is given in an orthonormal frame by $\p'_j=e^k_j p_k'$, which is explicitly written as
\[
\p'_j=\p_j+2\bigg(-q^0\frac{\n_a\omega^a}{\sqrt{s}}+\q_a\omega^a+\frac{\n_a\omega^a\n_b\q^b}{\sqrt{s}(n^0+\sqrt{s})}\bigg)
\bigg(\eta_{aj}\omega^a+\frac{\n_a\omega^a\n_j}{\sqrt{s}(n^0+\sqrt{s})}\bigg).
\]
We use the estimate (27) of \cite{LeeNungesser172}, where the estimate of high order derivatives of $\p'$ is given by 
\[
|\hat{\partial}^A\p'|\leq C(q^0)^{|A|+4}
\]
for $|A|\geq 1$. We now apply the relation \eqref{relation} to the above. Since $p'_k=(e^{-1})^j_k\p'_j$, where $e^{-1}$ is the inverse of the matrix $e$ and $\partial^A p'_k=(e^{-1})^j_ke^A_B \hat{\partial}^B\p'_j$, we have
\[
|\partial^Ap'_k|\leq C\max_{a,b}|(e^{-1})^a_b|(\max_{c,d}|e^c_d|)^{|A|}(q^0)^{|A|+4}
\]
for each $k\in\{1,2,3\}$ and $|A|\geq 1$. This completes the proof of the estimate for $p'_*$, and the estimate of $q'_*$ is given by the same arguments.
\end{proof}

\section{Estimates for the Boltzmann equation}\label{sec_boltzmann}
In this section we study the Boltzmann equation for a given metric. Let us assume that a metric $^{(4)}g=-dt^2+g$ is given and the spatial metric $g$ satisfies the properties of {\sf (A)} below. We show that the Boltzmann equation admits global-in-time solutions for small initial data. The following are the assumptions on the spatial metric $g$.\bigskip

\noindent{\sf (A)} {\bf Assumptions on the spatial metric.}
Let $\gamma=(\Lambda/3)^{1/2}$. There exists a constant $c_0>0$ such that
\[
\frac{1}{c_0}e^{-2\gamma t}|p_*|^2\leq g^{ab}p_ap_b\leq c_0 e^{-2\gamma t}|p_*|^2,
\]
for any $p_*$. The Hubble variable $H$ and the scaled shear $F$ satisfy
\[
H=\gamma+O(e^{-2\gamma t}),\quad  0<F<\frac14.
\]
Each component of the metric satisfies $|g_{ab}|\leq Ce^{2\gamma t}$ and $|g^{ab}|\leq Ce^{-2\gamma t}$. Let $e^a_b$ be an orthonormal frame with the inverse $e^{-1}$. We assume that they satisfy $|e^a_b|\leq C e^{\gamma t}$ and $|(e^{-1})^a_b|\leq Ce^{-\gamma t}$.
\bigskip

\begin{lemma}\label{lem_monotone}
For each $p_*$, the component $p^0$ is monotonically decreasing in $t$, i.e. $\partial_t p^0<0$, and can also be estimated as
\[
|\partial_t p^0|\leq C\langle p_*\rangle,
\]
where $C$ is independent of $t$.
\end{lemma}
\begin{proof}
On the mass shell we have $p^0=(1+g^{ab}p_ap_b)^{1/2}$ and obtain
\[
\partial_t p^0=-\frac{k^{ab}p_ap_b}{p^0}=-\frac{\sigma^{ab}p_ap_b+Hg^{ab}p_ap_b}{p^0}.
\]
The quantity $\sigma^{ab}p_ap_b$ can be estimated as
\[
|\sigma^{ab}p_ap_b|\leq (\sigma^{ab}\sigma_{ab})^{\frac12}(g^{cd}p_cp_d)=2HF^{\frac12}(g^{ab}p_ap_b),
\]
which shows that
\[
\partial_tp^0\leq \frac{H (-1+2F^{\frac12})(g^{ab}p_ap_b)}{p^0}.
\]
The right side is negative by the assumption {\sf(A)}, and this shows that $p^0$ is monotonically decreasing in $t$. With the above estimate of $\sigma^{ab}p_ap_b$ we also have
\[
|\partial_tp^0|\leq H(1+2F^{\frac12})p^0.
\]
Since $p^0\leq C\langle p_*\rangle$ by the assumption {\sf (A)}, we obtain the second estimate, and this completes the proof.
\end{proof}

\begin{lemma}\label{lem_f1}
Let $f$ be a solution of the Boltzmann equation. Then, it satisfies the following estimate:
\[
\|f(t)\|^2_k\leq \|f(t_0)\|^2_k+C\sup_{\tau\in[t_0,t]}\|f(\tau)\|^3_k,
\]
where $k$ is a positive real number.
\end{lemma}
\begin{proof}
The proof of this lemma is almost the same with that of Lemma 9 of \cite{LeeNungesser172}. Multiplying the Boltzmann equation by $f$ and integrating it on $[0,t]$, we obtain
\[
f^2(t,p_*)=f^2(t_0,p_*)+2\int_{t_0}^tf(\tau,p_*)Q(f,f)(\tau,p_*)d\tau.
\]
Note that the quantity $p^0(t)$ is monotonically decreasing in $t$ for each $p_*$ by Lemma \ref{lem_monotone}. We use this monotone property to estimate the above as follows:
\begin{align*}
&\langle p_*\rangle^{2k} e^{p^0(t)}f^2(t,p_*)\leq\langle p_*\rangle^{2k} e^{p^0(t_0)}f^2(t_0,p_*)\\
&\quad +C\int_{t_0}^t(\det g)^{-\frac12}\langle p_*\rangle^ke^{\frac12p^0(\tau)}f(\tau,p_*)\\
&\qquad\times\iint \frac{e^{-\frac12q^0(\tau)}}{p^0q^0\sqrt{s}}\langle p_*'\rangle^ke^{\frac12p'^0(\tau)}f(\tau,p_*')\langle q_*'\rangle^ke^{\frac12q'^0(\tau)}f(\tau,q_*')d\omega dq_* d\tau,
\end{align*}
where we ignored the loss term and used Lemma \ref{lem_pp'q'} with the energy conservation at time $\tau$:
\[
p'^0(\tau)+q'^0(\tau)=p^0(\tau)+q^0(\tau).
\]
Integrating the above with respect to $p_*$, we obtain the following estimate:
\begin{align}
\|f(t)\|_k^2&\leq\|f(t_0)\|^2_k+C\int_{t_0}^t(\det g)^{-\frac12}\int\langle p_*\rangle^ke^{\frac12p^0(\tau)}f(\tau,p_*)\nonumber\\
&\qquad\times\iint \frac{e^{-\frac12q^0(\tau)}}{p^0q^0\sqrt{s}}\langle p_*'\rangle^ke^{\frac12p'^0(\tau)}f(\tau,p_*')\langle q_*'\rangle^ke^{\frac12q'^0(\tau)}f(\tau,q_*')d\omega dq_* dp_*d\tau\allowdisplaybreaks\nonumber\\
&\leq\|f(t_0)\|^2_k+C\int_{t_0}^t(\det g)^{-\frac12}\bigg(\iiint\langle p_*\rangle^{2k}e^{p^0(\tau)}f^2(p_*)e^{-q^0(\tau)}d\omega dq_* dp_*\bigg)^{\frac12}\nonumber\\
&\qquad\times\bigg(\iiint \frac{1}{p^0q^0}\langle p_*'\rangle^{2k}e^{p'^0(\tau)}f^2(p_*')\langle q_*'\rangle^{2k}e^{q'^0(\tau)}f^2(q_*')d\omega dq_* dp_*\bigg)^{\frac12}d\tau\allowdisplaybreaks\nonumber\\
&\leq\|f(t_0)\|^2_k+C\int_{t_0}^t(\det g)^{-\frac12}\|f(\tau)\|_k\bigg(\int e^{-q^0(\tau)}dq_* \bigg)^{\frac12}\nonumber\\
&\qquad\times\bigg(\iiint \frac{1}{p^0q^0}\langle p_*\rangle^{2k}e^{p^0(\tau)}f^2(p_*)\langle q_*\rangle^{2k}e^{q^0(\tau)}f^2(q_*)d\omega dq_* dp_*\bigg)^{\frac12}d\tau\allowdisplaybreaks\nonumber\\
&\leq\|f(t_0)\|^2_k+C\int_{t_0}^t(\det g)^{-\frac14}\|f(\tau)\|_k^3d\tau,\label{est_f}
\end{align}
where we used $(p^0q^0)^{-1}dp_*dq_*=(p'^0q'^0)^{-1}dp_*'dq_*'$ and Lemma \ref{lem_int}. Since the quantity $(\det g)^{-1/4}$ is integrable, we obtain the desired result.
\end{proof}

\begin{lemma}\label{lem_f2}
Let $f$ be a solution of the Boltzmann equation. Then, it satisfies the following estimate:
\[
\|\partial^Af(t)\|^2_k\leq \|\partial^A f(t_0)\|^2_{k}+C\sup_{\tau\in[t_0,t]}\|f(\tau)\|_{k,N}^3,
\]
where $1\leq |A|\leq N$ is a multi-index and $k$ is a positive real number.
\end{lemma}
\begin{proof}
The proof of this lemma is also almost similar to that of Lemma 10 of \cite{LeeNungesser172}, and we briefly sketch the proof. For a multi-index $A\neq 0$, we take $\partial^A$ to the Boltzmann equation, multiply the equation by $\partial^A f$, and integrate it over $[t_0,t]$ to obtain
\begin{align*}
&(\partial^A f)^2(t,p_*)=(\partial^A f)^2(t_0,p_*)\\
&\quad +2\int_0^t(\det g)^{-\frac12}\partial^A f(\tau,p_*)\sum\iint \partial^{A_0}\bigg[\frac{1}{p^0q^0\sqrt{s}}\bigg] \partial^{A_1}\Big[f(p_*')\Big]\partial^{A_2}\Big[f(q_*')\Big]d\omega dq_*d\tau\\
&\quad -2\int_0^t(\det g)^{-\frac12}\partial^A f(\tau,p_*)\sum\iint \partial^{A_0}\bigg[\frac{1}{p^0q^0\sqrt{s}}\bigg] (\partial^{A_1}f)(p_*)f(q_*)d\omega dq_*d\tau,
\end{align*}
where the summations are taken for all the possible $A_0$, $A_1$, and $A_2$ satisfying $A=A_0+A_1+A_2$ and $A=A_0+A_1$, respectively. Note that the multi-index notation in this paper is different from the one used in \cite{LeeNungesser172}, and here $A=A_0+A_1$ means that the set $A$ is equal to the disjoint union $A_0\sqcup A_1$, and $A=A_0+A_1+A_2$ is understood in a similar way. Multiplying the above by $\langle p_*\rangle^{2k} e^{p^0(t)}$ and using the monotone property of $p^0(t)$ and Lemma \ref{lem_pp'q'}, we obtain the following:
\begin{align}
&\langle p_*\rangle^{2k} e^{p^0(t)}(\partial^A f)^2(t,p_*)\leq \langle p_*\rangle^{2k} e^{p^0(t_0)}(\partial^A f)^2(t_0,p_*)\nonumber\\
&\quad +C\sum\int_{t_0}^t(\det g)^{-\frac12}\langle p_*\rangle^{k} e^{\frac12 p^0(\tau)}|\partial^A f(\tau,p_*)|\nonumber\\
&\qquad\times\iint \bigg|\partial^{A_0}\bigg[\frac{1}{p^0q^0\sqrt{s}}\bigg]\bigg| e^{-\frac12 q^0(\tau)}\langle p_*'\rangle^{k} e^{\frac12 p'^0(\tau)}\Big|\partial^{A_1}\Big[f(p_*')\Big]\Big|\langle q_*'\rangle^{k} e^{\frac12 q'^0(\tau)}\Big|\partial^{A_2}\Big[f(q_*')\Big]\Big|d\omega dq_*d\tau\allowdisplaybreaks\nonumber\\
&\quad +C\sum\int_{t_0}^t(\det g)^{-\frac12}\langle p_*\rangle^{k} e^{\frac12 p^0(\tau)}|\partial^A f(\tau,p_*)|\nonumber\\
&\qquad\times\iint \bigg|\partial^{A_0}\bigg[\frac{1}{p^0q^0\sqrt{s}}\bigg]\bigg|e^{-\frac12 q^0(\tau)}\langle p_*\rangle^{k} e^{\frac12 p^0(\tau)}| (\partial^{A_1}f)(p_*)|\langle q_*\rangle^{k} e^{\frac12 q^0(\tau)}f(q_*)d\omega dq_*d\tau.\label{est_partial_f1}
\end{align}
The partial derivatives in the integrands are estimated by Lemma \ref{lem_partial_1/p0}, \ref{lem_partial_1/n0+s}, and \ref{lem_p'}. We first notice that the assumption {\sf (A)} shows that Lemma \ref{lem_p'} implies
\[
|\partial^Bp'_*|+|\partial^Bq'_*|\leq C(q^0)^{|B|+4}
\]
for any multi-index $|B|\geq 1$. Since the components $e^a_b$ are bounded by the assumption {\sf (A)}, Lemma \ref{lem_partial_1/p0} implies that
\[
\bigg|\partial^{A_0}\bigg[\frac{1}{p^0q^0\sqrt{s}}\bigg]\bigg|\leq \frac{C(q^0)^{|A_0|}}{p^0q^0}.
\]
The quantities $\partial^{A_1}[f(p_*')]$ and $\partial^{A_2}[f(q_*')]$ are estimated as in \cite{LeeNungesser172}. Applying Faa di Bruno's formula and Lemma \ref{lem_p'} with the assumption {\sf (A)} we obtain
\begin{align*}
\Big|\partial^{A_1}\Big[f(p_*')\Big]\Big|\leq C(q^0)^{5|A_1|}\sum |(\partial^{B}f)(p_*')|,
\end{align*}
where the summation is finite and taken over $1\leq|B|\leq |A_1|$. We obtain a similar estimate for $\partial^{A_2}[f(q_*')]$, and the inequality \eqref{est_partial_f1} is estimated as follows:
\begin{align*}
&\langle p_*\rangle^{2k} e^{p^0(t)}(\partial^A f)^2(t,p_*)\leq \langle p_*\rangle^{2k} e^{p^0(t_0)}(\partial^A f)^2(t_0,p_*)\nonumber\\
&\quad +C\sum\int_{t_0}^t(\det g)^{-\frac12}\langle p_*\rangle^{k} e^{\frac12 p^0(\tau)}|\partial^A f(\tau,p_*)|\nonumber\\
&\qquad\times\iint \frac{(q^0)^{5|A|}}{p^0q^0} e^{-\frac12 q^0(\tau)}\langle p_*'\rangle^{k} e^{\frac12 p'^0(\tau)}|(\partial^{B_1}f)(p_*')|\langle q_*'\rangle^{k} e^{\frac12 q'^0(\tau)}|(\partial^{B_2}f)(q_*')|d\omega dq_*d\tau\allowdisplaybreaks\nonumber\\
&\quad +C\sum\int_{t_0}^t(\det g)^{-\frac12}\langle p_*\rangle^{k} e^{\frac12 p^0(\tau)}|\partial^A f(\tau,p_*)|\nonumber\\
&\qquad\times\iint \frac{(q^0)^{|A_0|}}{p^0q^0}e^{-\frac12 q^0(\tau)}\langle p_*\rangle^{k} e^{\frac12 p^0(\tau)}| (\partial^{A_1}f)(p_*)|\langle q_*\rangle^{k} e^{\frac12 q^0(\tau)}f(q_*)d\omega dq_*d\tau,
\end{align*}
where the summation of the second term is taken over some $B_1$ and $B_2$ satisfying $|B_1|+|B_2|\leq |A|$. Integrating the above with respect to $p_*$, we obtain 
\begin{align*}
&\|\partial^A f(t)\|^2_k\leq \|\partial^Af(t_0)\|^2_k
+C\sum\int_{t_0}^t(\det g)^{-\frac12}\int\langle p_*\rangle^{k} e^{\frac12 p^0(\tau)}|\partial^A f(\tau,p_*)|\\
&\qquad\times\iint \frac{(q^0)^{5|A|}}{p^0q^0} e^{-\frac12 q^0(\tau)}\langle p_*'\rangle^{k} e^{\frac12 p'^0(\tau)}|(\partial^{B_1}f)(p_*')|\langle q_*'\rangle^{k} e^{\frac12 q'^0(\tau)}|(\partial^{B_2}f)(q_*')|d\omega dq_*dp_*d\tau\allowdisplaybreaks\\
&\quad +C\sum\int_{t_0}^t(\det g)^{-\frac12}\int\langle p_*\rangle^{k} e^{\frac12 p^0(\tau)}|\partial^A f(\tau,p_*)|\nonumber\\
&\qquad\times\iint \frac{(q^0)^{|A_0|}}{p^0q^0}e^{-\frac12 q^0(\tau)}\langle p_*\rangle^{k} e^{\frac12 p^0(\tau)}| (\partial^{A_1}f)(p_*)|\langle q_*\rangle^{k} e^{\frac12 q^0(\tau)}f(q_*)d\omega dq_*dp_*d\tau\allowdisplaybreaks\\
&\leq \|\partial^Af(t_0)\|^2_k 
+C\sum\int_{t_0}^t(\det g)^{-\frac12}\bigg(\iiint \langle p_*\rangle^{2k} e^{p^0(\tau)}|\partial^A f(p_*)|^2(q^0)^{10|A|}e^{-q^0(\tau)}d\omega dq_* dp_*\bigg)^{\frac12}\\
&\qquad\times\bigg(\iiint \frac{1}{p^0q^0} \langle p_*'\rangle^{2k} e^{p'^0(\tau)}|(\partial^{B_1}f)(p_*')|^2\langle q_*'\rangle^{2k} e^{q'^0(\tau)}|(\partial^{B_2}f)(q_*')|^2d\omega dq_*dp_*\bigg)^{\frac12}d\tau\allowdisplaybreaks\\
&\quad+C\sum\int_{t_0}^t(\det g)^{-\frac12}\bigg(\iiint \langle p_*\rangle^{2k} e^{p^0(\tau)}|\partial^A f(p_*)|^2(q^0)^{2|A_0|}e^{-q^0(\tau)}d\omega dq_* dp_*\bigg)^{\frac12}\\
&\qquad\times\bigg(\iiint \frac{1}{p^0q^0} \langle p_*\rangle^{2k} e^{p^0(\tau)}|(\partial^{A_1}f)(p_*)|^2\langle q_*\rangle^{2k} e^{q^0(\tau)}f^2(q_*)d\omega dq_*dp_*\bigg)^{\frac12}d\tau\allowdisplaybreaks\\
&\leq \|\partial^Af(t_0)\|^2_k 
+C\sum\int_{t_0}^t(\det g)^{-\frac12}\|\partial^Af(\tau)\|_k\bigg(\int (q^0)^{10|A|}e^{-q^0(\tau)}dq_*\bigg)^{\frac12}\\
&\qquad\times\bigg(\iiint \frac{1}{p^0q^0} \langle p_*\rangle^{2k} e^{p^0(\tau)}|(\partial^{B_1}f)(p_*)|^2\langle q_*\rangle^{2k} e^{q^0(\tau)}|(\partial^{B_2}f)(q_*)|^2d\omega dq_*dp_*\bigg)^{\frac12}d\tau\allowdisplaybreaks\\
&\quad+C\sum\int_{t_0}^t(\det g)^{-\frac12}\|\partial^Af(\tau)\|_k\bigg(\int (q^0)^{2|A_0|}e^{-q^0(\tau)}dq_* \bigg)^{\frac12}\|\partial^{A_1}f(\tau)\|_k\|f(\tau)\|_kd\tau\allowdisplaybreaks\\
&\leq \|\partial^Af(t_0)\|^2_k 
+C\sum\int_{t_0}^t(\det g)^{-\frac14}\|\partial^Af(\tau)\|_k\|\partial^{B_1}f(\tau)\|_k\|\partial^{B_2}f(\tau)\|_kd\tau\allowdisplaybreaks\\
&\quad+C\sum\int_{t_0}^t(\det g)^{-\frac14}\|\partial^Af(\tau)\|_k\|\partial^{A_1}f(\tau)\|_k\|f(\tau)\|_kd\tau,
\end{align*}
where we used $(p^0q^0)^{-1}dp_*dq_*=(p'^0q'^0)^{-1}dp_*'dq_*'$ and Lemma \ref{lem_int}. We obtain
\begin{align}\label{est_partial_f2}
\|\partial^A f(t)\|^2_k\leq \|\partial^Af(t_0)\|^2_k+C\int_{t_0}^t(\det g)^{-\frac14}\|f(\tau)\|^3_{k,N}d\tau,
\end{align}
and the integrability of $(\det g)^{-1/4}$ gives the desired result.
\end{proof}

\subsection{Global-in-time existence for the Boltzmann equation}\label{sec_boltzmann_existence}
For a given metric $g$ satisfying the assumption {\sf (A)} the local-in-time existence of classical solutions to the Boltzmann equation is obtained by a standard iteration method. The estimates of Lemma \ref{lem_f1} and \ref{lem_f2} show that small solutions are bounded globally in time such that
\[
\|f(t)\|^2_{k,N}\leq \|f(t_0)\|^2_{k,N}+C\sup_{\tau\in[t_0,t]}\|f(\tau)\|_{k,N}^3.
\]
Hence, we conclude that there exists a small $\varepsilon>0$ such that if initial data is given as $\|f(t_0)\|^2_{k,N}<\varepsilon$, then the corresponding solution exists globally in time and is bounded such that
\begin{align}\label{norm_bounded}
\sup_{t\in[t_0,\infty)}\|f(t)\|^2_{k,N}\leq C\varepsilon.
\end{align}
To ensure that $f$ is differentiable with respect to $t$ we consider again the estimates \eqref{est_f} and \eqref{est_partial_f2}. For a multi-index $A$ we have
\[
\partial_t\Big[\langle p_*\rangle^{2k}e^{p^0}(\partial^Af)^2\Big]=\langle p_*\rangle^{2k}(\partial_tp^0)e^{p^0}(\partial^Af)^2+2\langle p_*\rangle^{2k}e^{p^0}(\partial^Af)(\partial^AQ)(f,f).
\]
We integrate the above with respect to $p_*$ and use the estimate $|\partial_tp^0|\leq C\langle p_*\rangle$ of Lemma \ref{lem_monotone} to estimate the first quantity on the right side. The estimate of the second quantity is the same as in Lemma \ref{lem_f2}. Collecting all $|A|\leq N$, we obtain the following:
\[
\bigg|\frac{d}{dt}\|f(t)\|^2_{k,N}\bigg|\leq \|f(t)\|^2_{k+\frac12,N}+C(\det g)^{-\frac14}\|f(t)\|^3_{k,N}.
\]
Since \eqref{norm_bounded} holds for any $k$, the quantity $\|f(t)\|^2_{k+1/2,N}$ is also bounded globally in time. The right hand side of the above inequality is bounded, and this shows that the solution $f$ is continuous and also differentiable with respect to $t$ by the equation \eqref{boltzmann}. Uniqueness is easily proved by taking two solutions $f$ and $g$ with $f(t_0)=g(t_0)$ and applying Gr{\"o}nwall's inequality. We obtain the following result.

\begin{prop}\label{prop_boltzmann}
Suppose that a spatial metric $g$ satisfies the assumption {\sf (A)}. Then, there exists a small $\varepsilon>0$ such that if initial data is given as $\|f(t_0)\|_{k+1/2,N}<\varepsilon$ for $N\geq 3$, then there exists a unique classical solution of the Boltzmann equation \eqref{boltzmann} which exists globally in time and satisfies 
\[
\sup_{t\in[t_0,\infty)}\|f(t)\|^2_{k,N}\leq C\varepsilon.
\]
\end{prop}

\section{Proof of the main theorem and outlook}
We can now prove global-in-time existence of classical solutions to the Einstein-Boltzmann
system  \eqref{evolution1}--\eqref{constraint2} with a positive cosmological constant $\Lambda$. Suppose that initial data $g_{ab}(t_0)$, $k_{ab}(t_0)$, and $f(t_0)$ are given such that $H(t_0)<(7/6)^{1/2}\gamma$, and define an iteration for $\{g_n\}$, $\{k_n\}$, and $\{f_n\}$ as follows. Let $(g_0)_{ab}(t)=e^{2\gamma t} \bar{g}_{ab}(t_0)$ and $(k_0)_{ab}(t)=k_{ab}(t_0)$. Choose an orthonormal frame $(e_0)^a_b$ satisfying $(g_0)_{ab}=(e_0)^c_a(e_0)^d_b\eta_{cd}$, which is given by $(e_0)^a_b(t)=e^{\gamma t} \bar{e}^a_b(t_0)$, and let $(e_0^{-1})^a_b$ be the inverse of $(e_0)^a_b$. Then, $g_0$ satisfies the assumption {\sf (A)} of Section \ref{sec_boltzmann}. By Proposition \ref{prop_boltzmann}, there exists a small positive constant $\varepsilon$ such that if $\|f(t_0)\|_{k+1/2,N}<\varepsilon$, then there exists a unique classical solution $f_0$, which is the solution of the Boltzmann equation in a given spacetime with the metric $g_0$ and satisfies $\sup_{t\in[t_0,\infty)}\|f_0(t)\|^2_{k,N}\leq C\varepsilon$. Now, suppose that $f_n$ is given such that $\sup_{t\in[t_0,\infty)}\|f_n(t)\|^2_{k,N}\leq C\varepsilon$ with $\|f(t_0)\|_{k+1/2,N}<\varepsilon$. Then, we have
\[
\hat{f}_n(t,\hat{p})\leq C \varepsilon (1+ e^{2\gamma t} |\hat{p}|^2)^{-\frac12 k} e^{-\frac12 p^0},
\]
and applying Proposition \ref{prop_einstein} we obtain $g_{n+1}$ and $k_{n+1}$, which are the solutions of ODEs, which result when $g$ and $k$ of \eqref{evolution1}--\eqref{evolution2} are replaced by $g_{n+1}$ and $k_{n+1}$, respectively, and $\rho$ and $S_{ab}$ are constructed with $f_n$. It is clear that $g_{n+1}$ and $k_{n+1}$ satisfy the assumption {\sf (A)}, and appyling Proposition \ref{prop_boltzmann} again we obtain $f_{n+1}$, and this completes the iteration. The estimates of Propositions \ref{prop_einstein} and \ref{prop_boltzmann} show that the constructed quantities are uniformly bounded with the desired asymptotic behaviour, and taking the limit, up to a subsequence, we find classical functions $g$, $k$, and $f$. We have seen from Lemmas \ref{lem_F} and \ref{lem_monotone} that $p^0$ decays monotonically. As a consequence $p^0$ is bounded from above which gives us the future geodesic completeness. For more details we refer to \cite{LeeNungesser171}.

We thus have obtained the global existence and asymptotic behaviour of solutions to the Einstein-Boltzmann system, which extend the results of \cite{LeeNungesser171, LeeNungesser172}. A natural generalisation would be to consider higher Bianchi types. The isotropic spacetime with spatially flat topology and Bianchi I spacetimes are in fact the simplest in the sense that the Vlasov part is particularly simple. Thus, it would be of interest to extend \cite{LeeNungesser171, LeeNungesser172} to an FLRW spacetime with negative curvature with or without a cosmological constant. The scattering kernel considered here is physically well-motivated, however for simplicity we assumed that it does not depend on the scattering angle. A generalisation would thus be to remove this restriction. Similarly it is desirable to remove the smallness assumption and obtain a large data result as in \cite{ND}. Finally based on the work of Tod \cite{Tod} it is of interest to study this system with singular initial data, and these topics will be our future projects.

\section*{Acknowledgements}
H. Lee has been supported by the TJ Park Science Fellowship of POSCO TJ Park Foundation. This research was supported by Basic Science Research Program through the National Research Foundation of Korea (NRF) funded by the Ministry of Science, ICT \& Future Planning (NRF-2015R1C1A1A01055216). E.N. is currently funded by a Juan de la Cierva research fellowship from the Spanish government and this work has been partially supported by ICMAT Severo Ochoa project SEV-2015-0554 (MINECO).

\end{document}